\documentclass[11pt]{article}%
\usepackage{amsfonts}
\usepackage{amsmath}
\usepackage{amssymb}
\usepackage{amsthm}
\usepackage{graphicx}
\usepackage[margin=1in]{geometry}%
\newtheorem{theorem}{Theorem}

\newtheorem{algorithm}{Algorithm}

\newtheorem{definition}{Definition}

\begin{document}

\title{\textbf{The nonlinear Bernstein-Schr\"{o}dinger equation in Economics}}
\author{Alfred Galichon\thanks{\textit{\textit{Economics Department, }Sciences Po,
Paris, and New York University. Email: alfred.galichon@sciencespo.fr. Galichon
gratefully acknowledges funding from the European Research Council under the
European Union's Seventh Framework Programme (FP7/2007-2013) / ERC grant
agreements no 313699 and 295298, and FiME.} }
\and Scott Duke Kominers\thanks{\textit{Society of Fellows, Harvard University.
Email: kominers@fas.harvard.edu. Kominers gratefully acknowledges the support
of NSF grant CCF-1216095 and the Harvard Milton Fund. }}
\and Simon Weber\thanks{\textit{Economics Department, Sciences Po, Paris.}
\textit{Address: Sciences Po, Department of Economics, 27 rue Saint-Guillaume,
75007 Paris, France. Email: simon.weber@sciencespo.fr.}}}
\maketitle

\bigskip

\bigskip

\bigskip

In this note, we will review and extend some results from our previous work
\cite{GKW} where we introduced a novel approach to imperfectly transferable
utility and unobserved heterogeneity in tastes, based on a nonlinear
generalization of the Bernstein-Schr\"{o}dinger equation. We consider an
assignment problem where agents from two distinct populations may form pairs,
which generates utility to each agent. Utility may be transfered across
partners, possibly with frictions. This general framework hence encompasses
both the classic Non-Tranferable Utility model (NTU) of Gale and Shapley
\cite{GaleShapley62}, sometimes called the \textquotedblleft stable marriage
problem\textquotedblright, where there exists no technology to allow transfers
between matched partners; and the Transferable Utility (TU) model of Becker
\cite{Becker73} and Shapley-Shubik \cite{SS}, a.k.a. \textquotedblleft optimal
assignment problem,\textquotedblright\ where utility (money) is additively
transferable across partners.

If the NTU assumption seems natural for many markets (including school
choices), TU models are more appropriate in most settings where there can be
bargaining (labour and marriage markets for example). However, even in those
markets, there can be transfer frictions. For example, in marriage markets,
the transfers between partners might take the form of favor exchange (rather
than cash), and the cost of a favor to one partner may not exactly equal the
benefit to the other.

In \cite{GKW}, we thus developped a general Imperfectly Transferable Utility
model with unobserved heterogeneity, which includes as special cases the
classic fully- and non-transferable utility models, but also extends to
collective models, and settings with taxes on transfers, deadweight losses,
and risk aversion. As we argue in the present note, the models we consider
in~\cite{GKW} obey a particularly simple system of equations we dubbed
\textquotedblleft Nonlinear Bernstein-Schr\"{o}dinger
equation\textquotedblright. The present contribution present a general result
for the latter equation, and also derives several consequences. The main
result is derived in Section 1. Section 2 and Section 3 consider equilibrium
assignment problems with and without heterogeneity. Finally, we provide a
discussion of our results in Section 4.

\section{The main result}

For $x\in\mathcal{X}$ and $y\in\mathcal{Y}$, we consider a function
$M_{xy}:\mathbb{R}^{2}\rightarrow\mathbb{R}$, and $M_{x0}:\mathbb{R}%
\rightarrow\mathbb{R}$ and $M_{0y}:\mathbb{R}\rightarrow\mathbb{R}$. Let
$\left(  n_{x}\right)  _{x\in\mathcal{X}}$ and $\left(  m_{y}\right)
_{y\in\mathcal{Y}}$ be vectors of positive numbers. Consider the
\emph{nonlinear Bernstein-Schr\"{o}dinger system}, which consists in looking
for two vectors $u\in\mathbb{R}^{\mathcal{X}}$ and $v\in\mathbb{R}%
^{\mathcal{Y}}$ such that%
\[
\left\{
\begin{array}
[c]{c}%
M_{x0}\left(  u_{x}\right)  +\sum_{y\in\mathcal{Y}}M_{xy}\left(  u_{x}%
,v_{y}\right)  =n_{x}\\
M_{0y}\left(  v_{y}\right)  +\sum_{x\in\mathcal{X}}M_{xy}\left(  u_{x}%
,v_{y}\right)  =m_{y}%
\end{array}
\right.  .
\]

\begin{theorem}
\label{th:Existence}Assume $M_{xy}$ satisfies the following three conditions:

(i) Continuity. The maps $M_{xy}:(u_{x},v_{y})\longmapsto M_{xy}(u_{x}%
,v_{y}),$ $M_{x0}:(u_{x})\longmapsto M_{x0}(u_{x})$ and $M_{0y}:(v_{y}%
)\longmapsto M_{0y}(v_{y})$ are continuous

(ii) Monotonicity. The map $M_{xy}:(u_{x},v_{y})\longmapsto M_{xy}(u_{x}%
,v_{y})$ is monotonically decreasing, i.e if $u_{x}\leq u_{x}^{\prime}$ and
$u_{y}\leq u_{y}^{\prime}$, then $M_{xy}(u_{x},v_{y})\geq M(u_{x}^{\prime
},v_{y}^{\prime})$. The maps $M_{x0}:(u_{x})\longmapsto M_{x0}(u_{x})$ and
$M_{0y}:(v_{y})\longmapsto M_{0y}(v_{y})$ are monotonically decreasing.

(iii) Limits. For each $v_{y},\lim_{u_{x}\rightarrow\infty}M_{xy}(u_{x}%
,v_{y})=0$ and $\lim_{u_{x}\rightarrow-\infty}M_{xy}(u_{x},v_{y})=+\infty$,
and for each $u_{x},$ $\lim_{u_{y}\rightarrow\infty}M_{xy}(u_{x},v_{y})=0$ and
$\lim_{u_{y}\rightarrow-\infty}M_{xy}(u_{x},v_{y})=+\infty$. Additionally,
$\lim_{u_{x\rightarrow\infty}}M_{x0}(u_{x})=0$ and $\lim_{u_{y}\rightarrow
\infty}M_{0y}(u_{y})=0$, and $\lim_{u_{x}\rightarrow-\infty}M_{x0}%
(u_{x})=+\infty$ and $\lim_{u_{y}\rightarrow-\infty}M_{0y}(u_{y})=+\infty.$

\medskip

Then there exists a solution to the nonlinear Bernstein-Schr\"{o}dinger
system. Further, if the maps $M$ are $C^{1}$, the solution is unique.
\end{theorem}

\bigskip

This theorem appears in \cite{GKW} under a slightly different form. The proof
is interesting as it provides an algorithm of determination of $u$ and $v$. It
is an important generalization of the Iterated Projection Fitting Procedure
(see \cite{DS}, \cite{R}, and \cite{RT}), which has been rediscovered and
utilized many times under different names for various applied purposes:
\textquotedblleft RAS\ algorithm\textquotedblright\ \cite{K},
\textquotedblleft biproportional fitting\textquotedblright, \textquotedblleft
Sinkhorn Scaling\textquotedblright\ \cite{C}, etc. However, all these
techniques and their variants can be recast as particular cases of the method
described in the proof of Theorem~\ref{th:Existence}. For convenience, we
recall the algorithm used to provide a constructive proof of existence.

\begin{algorithm}
~

\label{algorithma}
\begin{tabular}
[c]{r|p{5in}}%
Step $0$ & Fix the initial value of $v_{y}$ at $v_{y}^{0}=+\infty$.\\
Step $2t+1$ & Keep the values $v_{y}^{2t}$ fixed. For each $x\in\mathcal{X}$,
solve for the value, $u_{x}^{2t+1}$ such that equality $\sum_{y\in\mathcal{Y}%
}M_{xy}(u_{x},v_{y}^{2t})+M_{x0}(u_{x})=n_{x}$ holds.\\
Step $2t+2$ & Keep the values $u_{x}^{2t+1}$ fixed. For each $y\in\mathcal{Y}%
$, solve for which is the value, $v_{y}^{2t+2}$ such that equality $\sum
_{x\in\mathcal{X}}M_{xy}(u_{x}^{2t+1},v_{y})+M_{0y}(v_{y})=m_{y}$ holds.
\end{tabular}
Then $u^{t}$ and $v^{t}$ converge monotonically to a solution of the
Bernstein-Schr\"{o}dinger system.
\end{algorithm}

In practice a precision level $\epsilon>0$ will be chosen, and the algorithm
will terminate when $\sup_{y}|v_{y}^{2t+2}-v_{y}^{2t}|<\epsilon$.

\begin{proof}
(i) Existence. The proof of existence is an application of Tarski's fixed
point theorem and relies on the previous Algorithm. We need to prove that the
construction of $u_{x}^{2t+1}$ and $v_{y}^{2t+2}$ at each step is well
defined. Consider step $2t+1$. For each $x\in\mathcal{X}$, the equation to
solve is
\[
\sum_{y\in\mathcal{Y}}M_{xy}(u_{x},v_{y})+M_{x0}(u_{x})=n_{x}%
\]
but the right handside is a continuous and decreasing function of $u_{x}$,
tends to $0$ when $u_{x}\rightarrow+\infty$ and tends to $+\infty$ when
$u_{x}\rightarrow-\infty$. Note that by letting $v_{y}\rightarrow+\infty$, the
terms in the sum tends to $0$, providing a lower bound for $u_{x}$. Hence
$u_{x}^{2t+1}$ is well defined and belongs in $\left(  M_{x0}^{-1}%
(n_{x}),+\infty\right)  $, and let us denote
\[
u_{x}^{2t+1}=F_{x}(v_{.}^{2t})
\]
and clearly, $F$ is anti-isotone, meaning that $v_{y}^{2t}\leq\tilde{v}%
_{y}^{2t}$ for all $y\in\mathcal{Y}$ implies $F_{x}(\tilde{v}_{.}^{2t})\leq
F_{x}(v_{.}^{2t})$ for all $x\in\mathcal{X}$. By the same token, at step
$2t+2$, $v_{y}^{2t+2}$ is well defined in $\left(  M_{0y}^{-1}(m_{y}%
),+\infty\right)  $, and let us denote%
\[
v_{y}^{2t+2}=G_{y}(u_{.}^{2t+1})
\]
where, similarly, $G$ is anti-isotone. Thus%
\[
v_{.}^{2t+2}=G\circ F\left(  v_{.}^{2t}\right)
\]
where $G\circ F$ is isotone. But $v_{y}^{2}<\infty=v_{y}^{0}$ implies that
$v_{.}^{2t+2}\leq G\circ F\left(  v_{.}^{2t}\right)  $. Hence $\left(
v_{.}^{2t+2}\right)  _{t\in\mathbb{N}}$ is a decreasing sequence, bounded from
below by $0$. As a result $v_{.}^{2t+2}$ converges. Letting $\bar{v}_{.}$ its
limit, and letting $\bar{u}=F(\bar{v})$, one can see that $(\bar{u},\bar{v})$
is a solution to the nonlinear \emph{Bernstein-Schr\"{o}dinger system.}
(ii) Unicity. Introduce map $\zeta$ defined by
\[
\zeta:\left(  u_{x},v_{y}\right)  \rightarrow\binom{\zeta_{x}=\sum
_{y\in\mathcal{Y}}M_{xy}\left(  u_{x},v_{y}\right)  +M_{x0}(u_{x})}{\zeta
_{y}=\sum_{x\in\mathcal{X}}M_{xy}\left(  u_{x},v_{y}\right)  +M_{0y}(v_{y})}%
\]
One has
\[
D\zeta\left(  u_{x},v_{y}\right)  =\left(
\begin{array}
[c]{c|c}%
A & B\\\hline
C & D
\end{array}
\right)
\]
where:
\begin{itemize}
\item $A=\left(  \partial\zeta_{x}/\partial u_{x^{\prime}}\right)
_{xx^{\prime}}=\sum_{y^{\prime}\in\mathcal{Y}}\partial_{u_{x}}M_{xy^{\prime}%
}\left(  u_{x},v_{y^{\prime}}\right)  +1$ if $x=x^{\prime}$, $0$ otherwise.
\item $B=\left(  \partial\zeta_{x}/\partial v_{y}\right)  _{xy}=\partial
_{v_{y}}M_{xy}\left(  u_{x},v_{y}\right)  $
\item $C=\left(  \partial\zeta_{y}/\partial u_{x}\right)  _{yx}=\partial
_{u_{x}}M_{xy}\left(  u_{x},v_{y}\right)  $
\item $D=\left(  \partial\zeta_{y}/\partial v_{y^{\prime}}\right)
_{yy^{\prime}}=\sum_{x^{\prime}\in\mathcal{X}}\partial_{v_{y}}M_{x^{\prime}%
y}\left(  u_{x^{\prime}}v_{y}\right)  +1$ if $y=y^{\prime}$, $0$ otherwise.
\end{itemize}
It is straightforward to show that the matrix $D\zeta$ is dominant diagonal. A
result from \cite{McK} states that a dominant diagonal matrix with positive
diagonal entries is a P-matrix. Hence $D\zeta\left(  u_{x},v_{y}\right)  $ is
a P-matrix. Applying Theorem 4 in \cite{GN} it follows that $\zeta$ is injective.
\end{proof}
\section{Equilibrium Assignment Problem}
In this section, we consider the equilibrium assignment problem, which is a
far-reaching generalization of the optimal assignment problem. To describe
this framework, consider two finite populations $\mathcal{I}$ and
$\mathcal{J}$, and a two-sided matching framework (for simplicity, we will
call \textquotedblleft men\textquotedblright\ and \textquotedblleft
women\textquotedblright\ the two sides of this market) with imperfect
transfers and without heterogeneity. Agents $i\in\mathcal{I}$ and
$j\in\mathcal{J}$ get respectively utility $u_{i}$ and $v_{j}$ they get at
equiliburium. If $i$ or $j$ remains unmatched, they get utility $0$; however,
if they match together, they may get any respective utilities $u_{i}$ and
$v_{j}$ such that the feasibility constraint is imposed%
\begin{equation}
\Psi_{ij}\left(  u_{i},v_{j}\right)  \leq0, \label{feasib}%
\end{equation}
where the transfer function $\Psi_{ij}$ is assumed to be continuous and
isotone with respect to its arguments. Note that at equilibrium, $u_{i}\geq0$
and $v_{j}\geq0$ as the agents always have the option to remain unassigned; by
the same token, if for any pair $i$, $j$ (matched or not), one cannot have a
strict inequality in (\ref{feasib}), otherwise $i$ and $j$ would have an
incentive to form a blocking pair, and achieve a higher payoff than their
equilibrium payoff. Thus the stability condition $\Psi_{ij}\left(  u_{i}%
,v_{j}\right)  \geq0,$ holds in general. Let $\mu_{ij}=1$ if $i$ and $j$ are
matched, and $0$ otherwise; we have therefore that $\mu_{ij}>0$ implies that
$\Psi_{ij}\left(  u_{i},v_{j}\right)  =0$. This allows us to define an
equilibrium outcome.
\begin{definition}
\label{def:EAP}The equilibrium assignment problem defined by $\Psi$ has an
equilibrium outcome $\left(  \mu_{ij},u_{i},v_{j}\right)  $ whenever the
following conditions are met:
(i) $\mu_{ij}\geq0$, $u_{i}\geq0$ and $v_{j}\geq0$
(ii) $\sum_{j}\mu_{ij}\leq1$ and $\sum_{i}\mu_{ij}\leq1$
(iii) $\Psi_{ij}\left(  u_{i},v_{j}\right)  \geq0$
(iv) $\mu_{ij}>0$ implies $\Psi_{ij}\left(  u_{i},v_{j}\right)  =0$.
\end{definition}
Note that, by the Birkhoff-von Neumann theorem, the existence of an
equilibrium in this problem leads to the existence of an equilibrium
satisfying the stronger integrality requirement $\mu_{ij}\in\left\{
0,1\right\}  $.
\bigskip
This general framework allow us to express the optimal assignment problem
(matching with Transferable Utility), as the case where%
\[
\Psi_{ij}\left(  u_{i},v_{j}\right)  =u_{i}+v_{j}-\Phi_{ij},
\]
while in the NTU case%
\[
\Psi_{ij}\left(  u_{i},v_{j}\right)  =\max(u_{i}-\alpha_{ij},v_{j}-\gamma
_{ij}).
\]
Other interesting cases are considered in~\cite{GKW}. For instance, the
\textit{Linearly Transferable Utility} (\textit{LTU}) model, where%
\[
\Psi_{ij}\left(  u_{i},v_{j}\right)  =\lambda_{ij}(u_{i}-\alpha_{ij}%
)+\zeta_{ij}(v_{j}-\gamma_{ij})
\]
with $\lambda_{ij},\zeta_{ij}>0$, and the \textit{Exponentially Transferable
Utility} (\textit{ETU}) model, in which $\Psi_{ij}$ takes the form
\[
\Psi_{ij}\left(  u_{i},v_{j}\right)  =\tau\log\left(  \frac{\exp(u_{i}%
/\tau)+\exp(v_{j}/\tau)}{2}\right) .
\]
In the ETU model, the parameter $\tau_{ij}$ is defined as the \textit{degree
of transferability}, since $\tau\rightarrow+\infty$ recovers the TU case and
$\tau\rightarrow0$ recovers the NTU framework.
\bigskip
\begin{theorem}
Assume $\Psi$ is such that:
\label{a11}(a) For any $x\in\mathcal{X}$ and $y\in\mathcal{Y}$, we have
$\Psi_{xy}\left(  \cdot,\cdot\right)  $ continuous.
\label{a12}(b) For any $x\in\mathcal{X}$, $y\in\mathcal{Y}$, $t\leq t^{\prime
}$ and $r\leq r^{\prime}$, we have $\Psi_{xy}\left(  t,r\right)  \leq\Psi
_{xy}\left(  t^{\prime},r^{\prime}\right)  $; furthermore, when $t<t^{\prime}$
and $r<r^{\prime}$, we have $\Psi_{xy}\left(  t,r\right)  <\Psi_{xy}\left(
t^{\prime},r^{\prime}\right)  $.
\label{a13}(c) For any sequence $\left(  t_{n},r_{n}\right)  $, if $\left(
r_{n}\right)  $ is bounded and $t_{n}\rightarrow+\infty$, then $\liminf
\Psi_{xy}\left(  t_{n},r_{n}\right)  >0$. Analogously, if $\left(
t_{n}\right)  $ is bounded and $r_{n}\rightarrow+\infty$, then $\liminf
\Psi_{xy}\left(  t_{n},r_{n}\right)  >0$.
\label{a14}(d) For any sequence $\left(  t_{n},r_{n}\right)  $ such that if
$\left(  t_{n}-r_{n}\right)  $ is bounded and $t_{n}\rightarrow-\infty$ (or
equivalently, $r_{n}\rightarrow+\infty$), we have that $\limsup\Psi
_{xy}\left(  t_{n},r_{n}\right)  <0$.
Then the equilibrium assignment problem defined by $\Psi_{ij}$ has an
equilibrium outcome.
\end{theorem}
\begin{proof}
Consider $T>0$ and let
\begin{align*}
M_{ij}\left(  u_{i},v_{j}\right)   &  =\exp\left(  -\frac{\Psi_{ij}\left(
u_{i},v_{j}\right)  }{T}\right)  \\
M_{i0}\left(  u_{i}\right)   &  =\exp\left(  -\frac{u_{i}}{T}\right)  \\
M_{0j}\left(  v_{j}\right)   &  =\exp\left(  -\frac{v_{j}}{T}\right)
\end{align*}
and consider the Bernstein-Schr\"{o}dinger system
\[
\left\{
\begin{array}
[c]{c}%
M_{i0}\left(  u_{i}\right)  +\sum_{j\in\mathcal{J}}M_{ij}\left(  u_{i}%
,v_{j}\right)  =1\\
M_{0j}\left(  v_{j}\right)  +\sum_{i\in\mathcal{I}}M_{ij}\left(  u_{i}%
,v_{j}\right)  =1
\end{array}
\right.  .
\]
We need to show that $M_{xy}(.,.)$, $M_{x0}(.)$ and $M_{0y}(.)$ satisfy the
properties stated in Theorem \ref{th:Existence}. It is straightforward to show
that conditions (i) and (ii)  in Theorem \ref{th:Existence} follow directly from assumptions (a) and (b) on
$\Psi.$ Moreover, letting $T\rightarrow0^{+}$, assumptions (c) and (d) together imply
that condition (iii) is satisfied. Hence we can apply Theorem~\ref{th:Existence}, and it follows
that a solution $u_{i}^{T},v_{j}^{T}$ to the system exists. Note that
$M_{i0}\left(  u_{i}^{T}\right)  \leq1$ and $M_{0j}\left(  v_{j}^{T}\right)
\leq1$ imply that $u_{i}^{T}\geq0$ and $v_{j}^{T}\geq0$. Now, consider the
sequence obtained by taking $T=k$, $k\in\mathbb{N}$. Then, up to a subsequence
extraction, we may assume $u_{i}^{k}\rightarrow\bar{u}_{i}\in\mathbb{R}%
^{+}\cup\left\{  +\infty\right\}  $ and $v_{j}^{k}\rightarrow\bar{v}_{j}%
\in\mathbb{R}^{+}\cup\left\{  +\infty\right\}  $. It follows that $\Psi
_{ij}\left(  u_{i}^{k},v_{j}^{k}\right)  $ converges in $\mathbb{R}^{+}%
\cup\left\{  +\infty\right\}  $, hence $\mu_{ij}^{k}=M_{ij}\left(  u_{i}%
^{k},v_{j}^{k}\right)  $ converges toward $\bar{\mu}_{ij}\in\left[
0,1\right]  $. Similarly, the limits $\bar{\mu}_{i0}$ and $\bar{\mu}_{0j}$
exist in $\left[  0,1\right]  $. Hence (i) in Definition~\ref{def:EAP} is met.
By continuity, $\bar{\mu}$ satisfies
\[
\left\{
\begin{array}
[c]{c}%
\bar{\mu}_{i0}+\sum_{j\in\mathcal{J}}\bar{\mu}_{ij}=1\\
\bar{\mu}_{0j}+\sum_{i\in\mathcal{I}}\bar{\mu}_{ij}=1
\end{array}
\right.  ,
\]
which established (ii). Let us show that (iii) holds, that is, that $\Psi
_{ij}\left(  u_{i},v_{j}\right)  \geq0$ for any $i$ and $j$. Assume otherwise. Then there exists $\epsilon>0$ such that for $k$ large enough,
$\Psi_{ij}\left(  u_{i}^{k},v_{j}^{k}\right)  <-\epsilon$, so that $\mu_{ij}%
^{k}>\exp\left(  \epsilon/T\right)  \rightarrow+\infty$, contradicting
$\bar{\mu}_{ij}\leq1$. Thus, we have established (iii). Finally, we show that (iv)
holds. Assume otherwise. Then there is $i$ and $j$ such that $\bar{\mu}%
_{ij}>0$ and $\Psi_{ij}\left(  \bar{u}_{i},\bar{v}_{j}\right)  >0$. This
implies that there exists $\epsilon>0$ such that for $k$ large enough,
$\mu_{ij}^{k}>\epsilon$, thus $\Psi_{ij}\left(  u_{i}^{k},v_{j}^{k}\right)
<-T\log\epsilon\rightarrow0$, hence $\Psi_{ij}\left(  \bar{u}_{i},\bar{v}%
_{j}\right)  \leq0$, a contradiction. Hence (iv) holds; this completes the
proof and establishes that $\left(  \bar{\mu},\bar{u},\bar{v}\right)  $ is an
equilibrium assignment.
\end{proof}

\section{ITU\ matching with heterogeneity}

Following \cite{GKW}, we now assume that individuals may be gathered into
groups of agents of similar observable characteristics, or types, but
heterogeneous tastes. We let $\mathcal{X}$ and $\mathcal{Y}$ be the sets of
\textit{types} of men and women. An individual man $i\in\mathcal{I}$ has type
$x_{i}\in\mathcal{X}$; similarly, an individual woman $j\in\mathcal{J}$ has
type $y_{j}\in\mathcal{Y}$. We assume that there is a mass $n_{x}$ of men of
type $x$ and $m_{y}$ of women of type $y$, respectively. Assume further that%
\[
\Psi_{ij}\left(  u_{i},v_{j}\right)  =\Psi_{x_{i}y_{j}}\left(  u_{i}%
-T\varepsilon_{iy},v_{j}-T\eta_{xj}\right)  ,
\]
where $\epsilon$ and $\eta$ are i.i.d.~random vectors drawn from a Gumbel
distribution, and where $T>0$ is a temperature parameter. Unassigned agents
get $T\varepsilon_{i0}$ and $T\eta_{0j}$. For all $i$ such that $x_{i}=x$ and
$y_{j}=y$, the stability condition implies%
\[
\Psi_{x_{i}y_{j}}\left(  u_{i}-T\varepsilon_{iy},v_{j}-T\eta_{xj}\right)
\geq0.
\]
Hence,%
\[
\min_{\substack{i:x_{i}=x\\j:y_{j}=y}}\Psi_{x_{i}y_{j}}\left(  u_{i}%
-T\varepsilon_{iy},v_{j}-T\eta_{xj}\right)  \geq0.
\]
Thus, letting%
\[
U_{xy}=\min_{i:x_{i}=x}\left\{  u_{i}-T\varepsilon_{iy}\right\}  \text{ and
}V_{xy}=\min_{j:y_{j}=y}\left\{  v_{j}-T\eta_{xj}\right\}  ,
\]
we have $\Psi_{xy}\left(  U_{xy},V_{xy}\right)  \geq0$, and with%
\[
\mu_{xy}=\sum_{\substack{i:x_{i}=x\\j:y_{j}=y}}\mu_{ij}%
\]
we have that $\mu_{xy}>0$ implies $\Psi_{xy}\left(  U_{xy},V_{xy}\right)  =0$,
and by a standard argument (the random vectors $\epsilon$ and $\eta$ are drawn
from distributions with full support, hence there will be at least a man $i$
of type $x$ and a woman $j$ of type $y$ such that $i$ prefers type $y$ and $j$
prefers type $x$, that is, $%
\mu
_{xy}>0$ for all $x$ and $y$)
\[
\Psi_{xy}\left(  U_{xy},V_{xy}\right)  =0~\forall x\in\mathcal{X}%
,y\in\mathcal{Y}.
\]
Note that this is an extension to the ITU\ case of the analysis in Galichon
and Salani\'{e} \cite{GS12}, building on Choo and Siow \cite{CS}. We have
that
\[
u_{i}=\max_{y}\left\{  U_{xy}+T\varepsilon_{iy},T\varepsilon_{i0}\right\}
\text{ and }v_{j}=\max_{x}\left\{  V_{xy}+T\eta_{xj},T\eta_{0j}\right\}  ,
\]
thus a standard result from Extreme Value Theory (see Choo and Siow \cite{CS}
for a derivation) yields
\[
U_{xy}=T\log\frac{\mu_{xy}}{\mu_{x0}}\text{ and }V_{xy}=T\log\frac{\mu_{xy}%
}{\mu_{0y}},
\]
so we see that $\mu_{xy}$ satisfies%
\begin{equation}
\left\{
\begin{array}
[c]{c}%
\mu_{x0}+\sum_{y\in\mathcal{Y}}\mu_{xy}=n_{x}\\
\mu_{0y}+\sum_{x\in\mathcal{X}}\mu_{xy}=m_{y}\\
\Psi_{xy}\left(  T\log\frac{\mu_{xy}}{\mu_{x0}},T\log\frac{\mu_{xy}}{\mu_{0y}%
}\right)  =0
\end{array}
\right.  .\label{aggrEq}%
\end{equation}
The various cases of interest discussed above, namely TU, NTU, LTU, and ETU
cases yield, respectively,
\begin{align*}
\mu_{xy}  & =\mu_{x0}^{1/2}\mu_{0y}^{1/2}\exp\frac{\Phi_{xy}}{2}\text{ (TU)}\\
\mu_{xy}  & =\min\left(  \mu_{x0}e^{\alpha_{xy}},\mu_{0y}e^{\gamma_{xy}%
}\right)  \text{ (NTU)}\\
\mu_{xy}  & =e^{\left(  \lambda_{xy}\alpha_{xy}+\zeta_{xy}\gamma_{xy}\right)
/\left(  \lambda_{xy}+\zeta_{xy}\right)  }\mu_{x0}^{\lambda_{xy}/\left(
\lambda_{xy}+\zeta_{xy}\right)  }\mu_{0y}^{\zeta_{xy}/\left(  \lambda
_{xy}+\zeta_{xy}\right)  }\text{ (LTU)}\\
\mu_{xy}  & =\left(  \frac{e^{-\alpha_{xy}/\tau_{xy}}\mu_{x0}^{-1/\tau_{xy}%
}+e^{-\gamma_{xy}/\tau_{xy}}\mu_{0y}^{-1/\tau_{xy}}}{2}\right)  ^{-\tau_{xy}%
}\text{ (ETU),}%
\end{align*}
(see \cite{CS}). To apply Theorem \ref{th:Existence}, we let $M_{xy}\left(
u_{x},v_{y}\right)  $ be the value $m$ solution to (for a proof of existence
and uniqueness of such a solution, see Lemma 1 of \cite{GKW})
\[
\Psi_{xy}\left(  T\log m+u_{x},T\log m+v_{y}\right)  =0,
\]
and let
\[
M_{x0}\left(  u_{x}\right)  =\exp\left(  \frac{-u_{x}}{T}\right)  \text{ and
}M_{0y}\left(  v_{y}\right)  =\exp\left(  \frac{-v_{y}}{T}\right)  .
\]
In \cite{GKW}, we rewrote system~(\ref{aggrEq}) as a nonlinear
Bernstein-Schr\"{o}dinger system, namely%
\[
\left\{
\begin{array}
[c]{c}%
M_{x0}\left(  u_{x}\right)  +\sum_{y\in\mathcal{Y}}M_{xy}\left(  u_{x}%
,v_{y}\right)  =n_{x}\\
M_{0y}\left(  v_{y}\right)  +\sum_{x\in\mathcal{X}}M_{xy}\left(  u_{x}%
,v_{y}\right)  =m_{y}%
\end{array}
\right.  .
\]

\begin{theorem}
The nonlinear Bernstein-Schr\"{o}dinger system in (\ref{aggrEq}) has a unique solution
\end{theorem}

\begin{proof}
The proof follows directly from the application of Theorem \ref{th:Existence}.
It is easy to check that the conditions on $M_{x0}(.)$ and $M_{0y}(.)$
required by Theorem \ref{th:Existence} are met in this case. Lemma 1 in \cite{GKW} provides a
proof that $M_{xy}$ satisfies these conditions.
\end{proof}

\section{Discussion}

In this note, we have argued how matching problems may be formulated as a
system of nonlinear equations, also known as the Bernstein-Schr\"{o}dinger
equation or Schr\"{o}dinger's problem \cite{S}. We have shown existence and
uniqueness of a solution under certain conditions, and have explicited the
link with various matching problems, with or without heterogeneity. Solving
such a system of equations requires an algorithm that we call the Iterative
Projection Fitting Procedure (IPFP); in practice, this algorithm converges
very quickly. Our setting can be extended in several ways. One of them is to
consider the case with unassigned agents. In that case, we have the additional
constraint that $\sum_{x}n_{x}=\sum_{y}m_{y}$, thus the nonlinear
Bernstein-Schr\"{o}dinger system, which in this case writes as%
\[
\left\{
\begin{array}
[c]{c}%
\sum_{y\in\mathcal{Y}}M_{xy}\left(  u_{x},v_{y}\right)  =n_{x}\\
\sum_{x\in\mathcal{X}}M_{xy}\left(  u_{x},v_{y}\right)  =m_{y}%
\end{array}
\right.
\]
has a degree of freedom, as the sum over $x\in\mathcal{X}$ of the first set of
equations coincides with the sum over $y\in\mathcal{Y}$ of the second one. The
one-dimensional manifold of solutions of this problem is studied in \cite{CG}.

\end{document}